\documentclass[11pt]{article}

\usepackage[margin=1in]{geometry}
\usepackage{amsmath,amssymb,amsthm} \usepackage{hhline}
\usepackage{bm}

\newtheorem{theorem}{Theorem}

\newtheorem{lemma}[theorem]{Lemma}
\newtheorem{proposition}[theorem]{Proposition}

\newtheorem*{theorem*}{Theorem} \newtheorem*{lemma*}{Lemma}

\theoremstyle{remark} \newtheorem*{remark*}{Remark}

\usepackage{todonotes}
\usepackage{hyperref}

\newcommand{\abs}[1]{\left\vert#1\right\vert}
\newcommand{\set}[1]{\left\{#1\right\}}
\newcommand{\tuple}[1]{\left(#1\right)} \newcommand{\eps}{\varepsilon}
 
\newcommand{\defeq}{\triangleq}

 \renewcommand{\mid}{\;\middle\vert\;}

\usepackage{xifthen}

\renewcommand{\Pr}[2][]{ \ifthenelse{\isempty{#1}}
  {\mathbf{Pr}\left[#2\right]} {\mathbf{Pr}_{#1}\left[#2\right]} }
\newcommand{\E}[2][]{ \ifthenelse{\isempty{#1}}
  {\mathop{\mathbf{E}}\left[#2\right]}
  {\mathop{\mathbf{E}}_{#1}\left[#2\right]} }

\begin{document}

\title{FPTAS for Hardcore and Ising Models on Hypergraphs }

\author{ Pinyan Lu\thanks{Microsoft Research. {\tt
      pinyanl@microsoft.com}} \and Kuan Yang\thanks{Shanghai Jiao Tong
    University. {\tt peanutyk@gmail.com}} \and Chihao
  Zhang\thanks{Shanghai Jiao Tong University. {\tt
      chihao.zhang@gmail.com}} } \date{}

\maketitle
\thispagestyle{empty}

\begin{abstract}
  Hardcore and Ising models are two most important families of two
  state spin systems in statistic physics. Partition function of spin
  systems is the center concept in statistic physics which connects
  microscopic particles and their interactions with their macroscopic
  and statistical properties of materials such as energy, entropy,
  ferromagnetism, etc. If each local interaction of the system
  involves only two particles, the system can be described by a
  graph. In this case, fully polynomial-time approximation scheme
  (FPTAS) for computing the partition function of both hardcore and
  anti-ferromagnetic Ising model was designed up to the uniqueness
  condition of the system. These result are the best possible since
  approximately computing the partition function beyond this threshold
  is NP-hard. In this paper, we generalize these results to general
  physics systems, where each local interaction may involves multiple
  particles. Such systems are described by hypergraphs. For hardcore
  model, we also provide FPTAS up to the uniqueness condition, and for
  anti-ferromagnetic Ising model, we obtain FPTAS where a slightly
  stronger condition holds.
\end{abstract}

\newpage
\setcounter{page}{1}

\section{Introduction}
In recent couple of years, there are remarkable progress on designing
approximate counting algorithms based on correlation decay
approach~\cite{Weitz06,BG08,GK07,RSTVY11,LLY12,SST,LLY13,LY13,counting-edge-cover,fibo-approx,LiuLZ14,LiuL15,LiuL15b}.
Unlike the previous major approximate counting approach that based on
random sampling such as Markov Chain Monte Carlo (MCMC) (see for
examples~\cite{MC_JA96,JS93, app_JSV04, app_GJ11, app_DJV01,
  col_Jerrum95, col_Vigoda99, IS_DFJ02, IS_DG00, IS_LV97}),
correlation decay based approach provides deterministic fully
polynomial-time approximation scheme (FPTAS).  New FPTASes were
designed for a number of interesting combinatorial counting problems
and computing partition functions for statistic physics systems, where partition
function is a weighted counting function from the computational point of
view. One most successful example is the algorithm for
anti-ferromagnetic two-spin systems~\cite{LLY12,SST,LLY13}, including
counting independent sets~\cite{Weitz06}. The correlation decay based
FPTAS is beyond the best known MCMC based FPRAS and achieves the
boundary of approximability~\cite{SS12,galanis2012inapproximability}.

In this paper, we generalize these results of anti-ferromagnetic
two-spin systems to hypergraphs. For physics point of view, this
corresponds to spin systems with higher order interactions, where each local interaction
involves more than two particles. There are two main ingredients for
the original algorithms and analysis on normal graphs (we will use the
term normal graph for a graph to emphasize that it is not hypergraphs):
(1) the construction of the self-avoiding walk tree by
Weitz~\cite{Weitz06}, which transform a general graph to a tree;
(2) correlation decay proof for the tree, which enables one to truncate
the tree to get a good approximation in polynomial time.  However, the
construction of the self-avoiding walk tree cannot be extended to
hypergraphs, which is the main obstacle for the generalization.

The most related previous work is counting independent sets for
hypergraphs by Liu and Lu~\cite{LiuL15}. They established a
computation tree with a two-layers recursive function instead of the
self-avoiding walk tree and provided a FPTAS to count the number of
independent sets for hypergraphs with maximum degree of 5, extending
the algorithm for normal graph with the same degree bound. Their proof
was significantly more complicated than the previous one due to the
complication of the two-layers recursive function. In particular, the
``right" degree bound for the problem is a real number between $5$ and
$6$ if one allow fraction degree in some sense. This integer gap
provides some room of flexibility and enables them to do some
case-by-case numerical argument to complete the proof. However, the
parameters for the anti-ferromagnetic two-spin systems on hypergraphs are
real numbers. To get a sharp threshold, we do not have any room for
numerical approximation.

\subsection{Our results}
We study two most important anti-ferromagnetic two-spin systems on
hypergraphs: the hardcore model and the anti-ferromagnetic Ising
model. The formal definitions of these two models can be found in
Section \ref{sec:prelim}.

Our first result is an FPTAS to compute the partition function of
hypergraph hardcore model.

\begin{theorem}\label{thm:ind-main}
  For hardcore model with a constant activity parameter of $\lambda$,
  there is an FPTAS to compute the partition function for hypergraphs
  with maximum degree $\Delta \geq 2 $ if $\lambda <
  \frac{(\Delta - 1)^{\Delta - 1}}{(\Delta - 2)^{\Delta}}$.
\end{theorem}

This bound is exactly the uniqueness threshold for the hardcore on
normal graphs. Thus, it is tight since normal graphs are special cases
of hypergraphs. To approximately compute the partition function beyond
this threshold is NP-hard.  In particular, The FPTAS in~\cite{LiuL15}
for counting the number of independent sets for hypergraphs with
maximum degree of $5$ can be viewed as a special case of our result
with parameters $\Delta=5, \lambda=1$, which satisfies the above
uniqueness condition. Another interesting special case is when
$\Delta=2$. This is not an interesting case for normal graphs since a
normal graph with maximum degree of $2$ is simply a disjoint union of
paths and cycles, whose partition function can be computed
exactly. However, the problem becomes more complicated on hypergraphs:
it can be interpreted as counting weighted edge covers on normal
graphs by viewing vertices of degree two as edges and hyperedges as
vertices. The exact counting of this problem is known to be
\#P-complete and an FPTAS was found recently~\cite{LiuLZ14}. In
our model, the uniqueness bound $\frac{(\Delta - 1)^{\Delta -
    1}}{(\Delta -2)^{\Delta}}$ is infinite for $\Delta=2$ and as a
result we give an FPTAS for counting weighted edge covers for any
constant edge weight $\lambda$. This gives an alternative proof for
the main result in~\cite{LiuLZ14}.

Our second result is on computing the partition function of
anti-ferromagnetic Ising model.

\begin{theorem}\label{thm:ising-main}
  For Ising model with interaction parameter $0 < \beta < 1$ and
  external field $\lambda$, there is an FPTAS to compute the partition
  function for hypergraphs with maximum degree $\Delta$ if $\beta \ge
  1 - \frac{2}{2e^{-1/2}\Delta+3}$.
\end{theorem}

The tight uniqueness bound for anti-ferromagnetic Ising model on
normal graphs is $\beta \ge 1 - \frac{2}{\Delta}$. So, our bound is
in the same asymptotic order but a bit worse in the constant
coefficient as $2 e^{-1/2}\approx 1.213 >1$. Moreover, our result can
apply beyond Ising model to a larger family of anti-ferromagnetic
two-spin systems on hypergraphs.

\subsection{Our techniques}
We also use the correlation decay approach. Although the framework of
this method is standard, in many work along this line of research, new
tools and techniques are developed to make this relatively new
approach more powerful and widely applicable. This is indeed the case
for the current paper as well. We summarize the new techniques we
introduced here.

For hardcore model, we replace the numerical case-by-case analysis by
a monotone argument with respect to the edge size of the hypergraph
which shows that the normal graphs with edge size of $2$ is indeed the
worst cases. This gives a tight bound for hardcore model.

To handle hypergraph with unbounded edge sizes, we need to prove that
the decay rate is much smaller for edges of larger size. Such effect
is called computationally efficient correlation decay, which has been
used in many previous works to obtain FPTASes for systems with
unbounded degrees or edge sizes. In all those works, one sets a
threshold for the parameter and proves different types of bounds for
large and small ones separately. Such artificial separation gets a
discontinuous bound which adds some complications in the proof and
usually ends with a case-by-case discussion. In particular, this
separation is not compatible with the above monotone argument. To
overcome this, we propose a new uniform and smooth treatment for this
by modifying the decay rate by a polynomial function of the edge
size. After this modification, we only need to prove one single bound
which automatically provides computationally efficient correlation. We
believe that this idea is important and may find applications in other
related problems.

For the Ising model, the main difficulty is to get a computation tree
as a replacement of the self-avoiding walk tree. We proposed one,
which also works for general anti-ferromagnetic two-spin systems on
hypergraphs. However, unlike the case of the hardcore model, the
computation tree is not of perfect efficiency and this is the main
reason that the bound we achieve in Theorem \ref{thm:ising-main} is
not tight. To get the computationally efficient correlation decay, we
also use the above mentioned uniform and smooth treatment. We also
extend our result beyond Ising to a family of anti-ferromagnetic
two-spin systems on hypergraphs.

\subsection{Discussion and open problems}
One obvious open question is to close the gap for Ising model, or more
generally extend our work to anti-ferromagnetic two-spin systems on
hypergraphs with better parameters. However, it seems that it is
impossible to obtain a tight result in these models using the
computation tree proposed in this paper, due to its imperfectness. 
How to overcome this is an important open question. 

Even for the hardcore model, our result is tight only for the family
of all hypergraphs since the normal graphs are special cases. From
both physics and combinatorics point of view, it would be very
interesting to study the family of $w$-uniform hypergraphs where each
hyperedge is of the same size $w$. By our monotone argument, it is
plausible to conjecture that one can get better bound for larger $w$.
In particular, MCMC based approach does show that larger edge size
helps: for hypergraph independent set with maximum degree of $\Delta$
and minimum edge size $w$, an FPRAS for $w \ge 2 \Delta + 1$ was shown
in~\cite{RSA:RSA20204}.  However, their result is not tight. Can we
get a tight bounds in terms of $\Delta$ and $w$ by correlation decay
approach? The high level idea sounds promising, but there is an
obstacle to prove such result by our computation tree. To construct
the computation tree, we need to construct modified instances. In
these modified instances, the size of a hyperedge may decrease to as
low as $2$. Therefore, even if we start with $w$-uniform hypergraphs
or hypergraphs with minimum edge size of $w$, we may need to handle
the worst case of normal graphs during the analysis. How to avoid this
effect is a major open question whose solution may have applications
in many other problems.

The fact that larger hyperedge size only makes the problem easier is not universally true for approximation counting. One interesting example is counting hypergraph matchings. FPTAS for counting 3D matchings of hypergraphs with maximum degree $4$ is given in~\cite{LiuL15}, and extension to weighted setting are studied in~\cite{yin2015counting}. In particular, a uniqueness condition in this setting is defined in~\cite{yin2015counting}, and it is a very interesting open question whether this uniqueness condition is also the transition boundary for approximability. 

%

\section{Preliminaries}\label{sec:prelim}

A hypergraph $G(V,\mathcal{E})$ consists of a vertex set $V$ and a set
of hyperedges $\mathcal{E}\subseteq 2^V$. For every hyperedge $e\in
\mathcal{E}$ and vertex $v\in V$, we use $e-v$ to denote
$e\setminus\set{v}$ and use $e+v$ to denote $e\cup\set{v}$.

\subsection{Hypergraph hardcore model}

The hardcore model is parameterized by the activity parameter
$\lambda>0$.  Let $G(V,\mathcal{E})$ be a hypergraph. An independent
set of $G$ is a vertex set $I\subseteq V$ such that $e\not\subseteq I$
for every hyperedge $e\in \mathcal{E}$. We use $\mathcal{I}(G)$ to
denote the set of independent sets of $G$. The weight of an
independent set $I$ is defined as $w(I)\defeq\lambda^{\abs{I}}$. We
let $Z(G)$ denote the partition function of $G(V,\mathcal{E})$ in the
hardcore model, which is defined as
\[
Z(G)\defeq\sum_{I\in\mathcal{I}(G)}w(I).
\]
The weight of independent sets induces a Gibbs measure on $G$. For
every $I\in\mathcal{I}$, we use
\[
\Pr[G]{I}\defeq\frac{w(I)}{Z(G)}
\]
to denote the probability of obtaining $I$ if we sample according to
the Gibbs measure. For every $v\in V$, we use
\[
\Pr[G]{v\in I}\defeq \sum_{\substack{I\in\mathcal{I}(G)\\v\in
    I}}\Pr[G]{I}
\]
to denote the \emph{marginal probability} of $v$.


\subsection{Hypergraph two state spin model}

Now we give a formal definition to hypergraph two state spin systems. This model is parameterized by the external
field $\lambda>0$. An instance of the model is a labeled hypergraph
$G(V,\mathcal{E},(\bm\beta,\bm\gamma))$ where
$\bm\beta,\bm\gamma:\mathcal{E}\to\mathbb{R}$ are two labeling
functions that assign each edge $e\in\mathcal{E}$ two reals
$\bm\beta(e),\bm\gamma(e)$. A configuration on $G$ is an assignment
$\sigma:V\to\set{0,1}$ whose weight $w(\sigma)$ is defined as
\[
w(\sigma)\defeq \prod_{e\in\mathcal{E}}w(e,\sigma)\prod_{v\in
    V}w(v,\sigma)
\]
where for a hyperedge $e=\set{v_1,\dots,v_w}$
\[
w(e,\sigma)\defeq
\begin{cases}
\bm\beta(e) & \mbox{if }\sigma(v_1)=\sigma(v_2)=\dots=\sigma(v_w)=0\\
\bm\gamma(e) & \mbox{if }\sigma(v_1)=\sigma(v_2)=\dots=\sigma(v_w)=1\\
1 & \mbox{otherwise}
\end{cases}
\]
and for a vertex $v$,
\[
w(v,\sigma)\defeq
\begin{cases}
\lambda & \mbox{if }\sigma(v)=1\\
1 & \mbox{otherwise.}
\end{cases}
\]

The partition function of the instance is given by
\[
Z(G) = \sum_{\sigma\in\set{0,1}^V} w(\sigma).
\]
Similarly, the weight of configurations induces a Gibbs measure on
$G$. For every $\sigma\in\set{0,1}^V$, we use
\[
\Pr[G]{\sigma}\defeq\frac{w(I)}{Z(G)}
\]
to denote the probability of $\sigma$ in the measure. For every $v\in
V$, we use
\[
\Pr[G]{\sigma(v)=1}\defeq
\sum_{\substack{\sigma\in\set{0,1}^V\\\sigma(v)=1}}\Pr[G]{\sigma}
\]
to denote the \emph{marginal probability} of $v$.

The \emph{anti-ferromagnetic} Ising model is the special case that
$\beta\defeq\bm\beta(e)=\bm\gamma(e)\le 1$ for all
$e\in\mathcal{E}$. In this model, we call $\beta$ the
\emph{interaction parameter} of the model. The hardcore model
introduced in previous section is the special case that
$\bm\beta(e)=1$ and $\bm\gamma(e)=0$ for all $e\in\mathcal{E}$.

We give the whole proof to Theorem \ref{thm:ising-main} in appendix. More precisely,
we design an FPTAS for the more general two state spin system and establish the following theorem:

\begin{theorem}\label{thm:twostate-main}
    Consider a class of two state spin system with external field
    $\lambda$ such that each instance
    $G(V,\mathcal{E},(\bm\beta,\bm\gamma))$ in the class satisfies
    $1-\frac{2}{2e^{-1/2}\Delta+3}\le\bm\beta(e),\bm\gamma(e)\le 1$ where
    $\Delta$ is the maximum degree of $G$. There exists an FPTAS to
    compute the partition function for every instance in the class.
\end{theorem}

Theorem \ref{thm:ising-main} then follows since it is a special case
of Theorem \ref{thm:twostate-main}.

Actually, the main idea of FPTAS design and proof for this model is similar to the idea we use to solve hypergraph hardcore model. However, the details of recursion function design and techniques for proof of correlation decay property are pretty different from that in hypergraph hardcore model, so we put the whole section in appendix.

\section{Hypergraph Hardcore Model}

\subsection{Recursion for computing marginal
  probability}\label{sec:ind-recursion}

We first fix some notations on graph modification specific to
hypergraph independent set. Let $G(V,\mathcal{E})$ be a hypergraph.
\begin{itemize}
\item For every $v\in V$, we denote $ G-v\defeq (V\setminus\set{v},
  \mathcal{E}') $ where $\mathcal{E}'\defeq\set{e\setminus\set{v}\mid
    e\in \mathcal{E}}$.
\item For every $e\in\mathcal{E}$, we denote $ G-e\defeq
  (V,\mathcal{E}\setminus\set{e}).  $
\item Let $x$ be a vertex or an edge and $y$ be a vertex or an edge,
  we denote $ G-x-y\defeq (G-x)-y.  $
\item Let $S=\set{v_1,\dots,v_k}\subseteq V$, we denote $ G-S\defeq
  G-v_1-v_2\cdots-v_k.  $
\item Let $\mathcal{F}=\set{e_1,\dots,e_k}\subseteq \mathcal{E}$, we
  denote $ G-\mathcal{F}\defeq G-e_1-e_2\cdots-e_k.$
\end{itemize}

Let $G(V,\mathcal{E})$ be a hypergraph and $v\in V$ be an arbitrary
vertex with degree $d$. Let $\set{e_1,\dots,e_d}$ be the set of
hyperedges incident to $v$ and for every $i\in[d]$, $e_i=\set{v}\cup
\set{v_{ij}\mid j\in [w_i]}$ consists of $w_i+1$ vertices.

We first define a graph $G'(V',\mathcal{E}')$, which is the graph
obtained from $G$ by replacing $v$ by $d$ copies of itself and each
$e_i$ contains a distinct copy. Formally,
$V'\defeq(V\setminus\set{v})\cup\set{v_1,\dots,v_d}$,
$\mathcal{E}'\defeq \set{e\in \mathcal{E}\mid v\not\in
  e}\cup\set{e_i-v+v_i\mid i\in[d]}$.

For every $i\in[d]$ and $j\in[w_i]$, we define a hypergraph
$G_{ij}(V_{ij},\mathcal{E}_{ij})$:
\[
G_{ij}\defeq G'-\set{v_k\mid i\le k\le d}-\set{e_k\mid 1\le k\le
  i}-\set{v_{ik}\mid 1\le k<j}.
\]



Let $R_v=\frac{\Pr[G]{v\in I}}{\Pr[G]{v\not\in I}}$ and
$R_{ij}=\frac{\Pr[G_{ij}]{v_{ij}\in I}}{\Pr[G_{ij}]{v_{ij}\not\in
    I}}$. We can compute $R_v$ by following recursion:
\begin{lemma}\label{lem:ind-recursion}
  \begin{equation}
    \label{eq:ind-recursion}
    R_v=\lambda\prod_{i=1}^d\tuple{1-\prod_{j=1}^{w_i}\frac{R_{ij}}{1+R_{ij}}}.
  \end{equation}
\end{lemma}

The proof of this lemma is postponed to appendix.

\paragraph{The Uniqueness Condition}
Let the underlying graph be an infinite $d$-ary tree, then the
recursion \eqref{eq:ind-recursion} becomes
\[
  f_{\lambda,d}(x)=\lambda \tuple{\frac{1}{1+x}}^d.
\]
Let $\hat x$ be the positive fixed-point of $f_{\lambda,d}(x)$, i.e.,
$\hat x>0$ and $f_{\lambda,d}(\hat x)=\hat x$. The condition on
$\lambda$ for the uniqueness of the Gibbs measure is that
$\abs{f'_{\lambda,d}(\hat x)}<1$.  The following proposition is
well-known.

\begin{proposition}
  Let $\lambda_c=\frac{d^d}{(d-1)^{d+1}}$, then
  $\abs{f'_{\lambda_c,d}(\hat x)}=1$ and for every
  $0<\lambda<\lambda_c$, it holds that
  $\abs{f'_{\lambda,d}(\hat x)}<1$.
\end{proposition}

\subsection{The algorithm to compute marginal
  probability}\label{sec:indset-algorithm}

Let $G(V,\mathcal{E})$ be a hypergraph with maximum degree $\Delta$
and $v\in V$ be an arbitrary vertex with degree $d$. Define $G_{ij}$,
$R_v$, $R_{ij}$ as in Section \ref{sec:ind-recursion}. Then the
recursion \eqref{eq:ind-recursion} gives a way to compute the marginal
probability $\Pr[G]{v\in I}$ \emph{exactly}. However, an exact
evaluation of the recursion requires a computation tree with
exponential size. Thus we introduce the following truncated version of
the recursion, with respect to constants $c>0$ and $0<\alpha<1$.
\[
  R(G,v,L) =
  \begin{cases}
    \lambda\prod_{i=1}^d\tuple{1-\prod_{j=1}^{w_i}\frac{R(G_{ij},v_{ij},L)}{1+R(G_{ij},v_{ij},L)}} & \mbox{if } d=\Delta\\
    \lambda\prod_{i=1}^d\tuple{1-\prod_{j=1}^{w_i}\frac{R(G_{ij},v_{ij},L-\lfloor
        1+c\log_{1/\alpha}
        w_i\rfloor)}{1+R(G_{ij},v_{ij},L-\lfloor 1+c\log_{1/\alpha} w_i\rfloor)}} & \mbox{if }d<\Delta\mbox{ and }L>0\\
    \lambda & \mbox{otherwise.}
  \end{cases}
\]

The recursion can be directly used to compute $R(G,v,L)$ for any given
$L$ and it induces a truncated computation tree (with height $L$ in
some special metric). It is worth noting that, the case that
$d=\Delta$ can only happen at the root of the computation tree, since
in each smaller instance, the degree of $v_{ij}$ is decreased by at
least one.

We claim that $R(G,v,L)$ is a good estimate of $R_v$ with a suitable
choice of $c$ and $\alpha$, for those $(\lambda,\Delta)$ in the
uniqueness region.

\begin{lemma}\label{lem:ind-decay}
  Let $G(V,\mathcal{E})$ be a hypergraph with maximum degree
  $\Delta\ge 2$. Let $v\in V$ be a vertex with degree $d$ and let
  $\lambda<\lambda_c=\frac{(\Delta-1)^{\Delta-1}}{(\Delta-2)^{\Delta}}$
  be the activity parameter.  There exist constants $C>0$ (more
  precisely, $C = 6\lambda\sqrt{1+\lambda}$) and $\alpha \in (0, 1)$
  such that
  \[
    \abs{R(G,v,L)-R_v}\le C\cdot\alpha^{\max\{0,L\}}
  \]
  for every $L$.
\end{lemma}

The whole proof of this lemma is postponed to the next section.

\begin{proof}[Proof of Theorem \ref{thm:ind-main}]
  Assuming Lemma \ref{lem:ind-decay}, the proof of Theorem
  \ref{thm:ind-main} is routine and we put the proof into appendix.
\end{proof}

\subsection{Correlation decay}

In this section, we establish Lemma \ref{lem:ind-decay}. We first prove
some technical lemmas.

Suppose $f : D^d \rightarrow \mathbb{R}$ is a $d$-ary function where $D
\subseteq \mathbb{R}$ is a convex set, let $\phi: \mathbb{R} \to \mathbb{R}$ be an
increasing differentiable function and $\Phi(x) \defeq \phi'(x)$.
The following proposition is a consequence of the mean value theorem:

\begin{proposition}\label{prop:tech}
  For every $\mathbf{x} = (x_1, \ldots , x_d)$, $\mathbf{\hat{x}} =
  (\hat{x}_1, \ldots , \hat{x}_d) \in D^d$, it holds that
  \begin{enumerate}
  \item $\abs{f(\mathbf{x}) - f(\mathbf{\hat{x}})} =
    \frac{1}{\Phi(\tilde{x})}\abs{\phi(f(\mathbf{x})) -
      \phi(f(\mathbf{\hat{x}}))}$ for some $\tilde{x} \in D$;
  \item $ \abs{\phi(f(\mathbf{x}))-\phi(f(\mathbf{\hat x}))} \le
    \sum_{i=1}^d\frac{\Phi(f)}{\Phi(\tilde x_i)}\abs{\frac{\partial
        f(\mathbf{\tilde x})}{\partial x_i}} \cdot
    \abs{\phi(x_i)-\phi(\hat x_i)} $ for some $\mathbf{\tilde
      x}=(\tilde x_1,\dots,\tilde x_d)\in D^d$.
  \end{enumerate}
\end{proposition}

\begin{lemma}\label{lem:derivative}
  Let $\Delta\ge 2$ be a constant integer and
  $\lambda<\lambda_c=\frac{(\Delta-1)^{\Delta-1}}{(\Delta-2)^{\Delta}}$
  be a constant real. Let $d<\Delta$ and $w_1,\dots,w_d>0$ be integers
  and
  $f=\lambda\prod_{i=1}^d\tuple{1-\prod_{j=1}^{w_i}\frac{x_{ij}}{1+x_{ij}}}$
  be a $\tuple{\sum_{i=1}^dw_i}$-ary function. Let
  $\Phi(x)=\frac{1}{\sqrt{x(1+x)}}$.
  Let $c<\min\left\{\frac{\log(1+\lambda)-\log\lambda}{2+4\lambda}, \frac{2\lambda+1}{2}\log\left(\frac{1+\lambda}{\lambda}\right)-1\right\}$ be a positive number. There exists
  a constant $\alpha<1$ depending on $\lambda$ and $d$ (but not
  depending on $w_i$ for all $i\in[d]$) such that
  \[
  \sum_{a=1}^dw_a^c\sum_{b=1}^{w_a}\frac{\Phi(f)}{\Phi(x_{ab})}\abs{\frac{\partial
      f(\mathbf{x})}{\partial x_{ab}}}\le \alpha<1
  \]
  for every ${\bf x}=(x_{ij})_{i\in[d],j\in[w_i]}$ where each
  $x_{ij}\in [0,\lambda]$
\end{lemma}

The lemma bounds the amortized decay rate, which is the key to the proof of correlation decay. 
In previous works, the amortized decay rate is defined as   \[
  \sum_{a=1}^d\sum_{b=1}^{w_a}\frac{\Phi(f)}{\Phi(x_{ab})}\abs{\frac{\partial
      f(\mathbf{x})}{\partial x_{ab}}},
  \]
  without the $w_a^c$ factor. Then one need to give a constant $\alpha<1$ bound for small $w_a$ and a sub constant bound for large  $w_a$. With this modification, we only need to prove a single bound as above. 

Notice that we require $c$ to be a positive constant, so it is necessary to verify that $\frac{2\lambda+1}{2}\log\left(\frac{1+\lambda}{\lambda}\right)-1 > 0$ for every $\lambda > 0$. To see this,
let
$h(\lambda)\defeq
\frac{2\lambda+1}{2}\log\tuple{\frac{1+\lambda}{\lambda}}-1$,
then we can compute that
\begin{align*}
h'(\lambda)&=\log\tuple{\frac{1+\lambda}{\lambda}}-\frac{1+2\lambda}{2\lambda+2\lambda^2},\\
h''(\lambda)&=\frac{1}{2\lambda^2(1+\lambda)^2}.
\end{align*}
Since $h''(\lambda)>0$ for every $\lambda$, $h'(\lambda)$ is
increasing. Along with the fact that
$\lim_{\lambda\to\infty}h'(\lambda)=0$, we have $h'(\lambda)<0$ for
every $\lambda> 0$. This implies that $h(\lambda)$ is decreasing. Also
note that
\[
\lim_{\lambda\to\infty}h(\lambda)=\lim_{\lambda\to\infty}\log\tuple{\tuple{1+\frac{1}{\lambda}}^\lambda\tuple{1+\frac{1}{\lambda}}^{1/2}}-1=0.
\]
It holds that $h(\lambda)>0$ for every $\lambda>0$. Thus a positive
$c$ satisfying $c<h(\lambda)$ exists for every $\lambda> 0$.

\begin{proof}[Proof of Lemma \ref{lem:derivative}]
  To simplify the notation, we first let
  $t_{ij}=\frac{x_{ij}}{1+x_{ij}}$, then for every $i\in[d]$ and
  $j\in[w_i]$, it holds that
  $t_{ij}\in\left[0,\frac{\lambda}{1+\lambda}\right]$ and
  \[
  f=\lambda\prod_{i=1}^d\tuple{1-\prod_{j=1}^{w_i}t_{ij}}.
  \]
  For every $a\in[d]$ and $b\in[w_i]$, we have
  \[
  \abs{\frac{\partial f}{\partial x_{ab}}}
  =\lambda(1-t_{ab})^2\prod_{\substack{j\in[w_a]\\j\ne
      b}}t_{aj}\cdot\prod_{\substack{i\in[d]\\i\ne
      a}}\tuple{1-\prod_{j=1}^{w_i}t_{ij}}
  =f\cdot\frac{(1-t_{ab})^2}{t_{ab}}\cdot\frac{\prod_{j=1}^{w_a}t_{aj}}{1-\prod_{j=1}^{w_a}t_{aj}}.
  \]
  Thus
  \[
  \sum_{a=1}^dw_a^c\sum_{b=1}^{w_a}\frac{\Phi(f)}{\Phi(x_{ab})}\abs{\frac{\partial
      f}{\partial x_{ab}}}=
  \sqrt{\frac{f}{1+f}}\sum_{a=1}^d\frac{w_a^c\prod_{j=1}^{w_a}t_{aj}}{1-\prod_{j=1}^{w_a}t_{aj}}\sum_{b=1}^{w_a}\frac{1-t_{ab}}
  {\sqrt{t_{ab}}}.
  \]
  Let $\mathbf{t}=(t_{ij})_{i\in[d],j\in [w_i]}$%
  , define
  \begin{align*}
    h(\mathbf{t})
    &\defeq\sqrt{\frac{f}{1+f}}\sum_{a=1}^d\frac{w_a^c\prod_{j=1}^{w_a}t_{aj}}{1-\prod_{j}^{w_a}t_{aj}}\sum_{b=1}^{w_a}\frac{1-t_{ab}}
    {\sqrt{t_{ab}}}\\
    &=\sqrt{\frac{\lambda\prod_{i=1}^d\tuple{1-\prod_{j=1}^{w_i}t_{ij}}}{1+\lambda\prod_{i=1}^d\tuple{1-\prod_{j=1}^{w_i}t_{ij}}}}\sum_{a=1}^d\frac{w_a^c\prod_{j=1}^{w_a}t_{aj}}{1-\prod_{j}^{w_a}t_{aj}}\sum_{b=1}^{w_a}\frac{1-t_{ab}}
    {\sqrt{t_{ab}}}.
  \end{align*}
  For every $\mathbf{t}=(t_{ij})_{i\in[d],j\in [w_i]}$ where each
  $t_{ij}\in[0,\frac{\lambda}{1+\lambda}]$, define a tuple
  $\mathbf{\hat t}=(\hat t_{ij})_{i\in[d],j\in w_i}$ such that for
  every $i\in [d]$,
  \[
  \hat t_{ij}=
  \begin{cases}
    \tuple{\frac{1+\lambda}{\lambda}}^{w_i-1}\prod_{k=1}^{w_i}t_{ik} & \mbox{if }j=1\\
    \frac{\lambda}{1+\lambda} & \mbox{otherwise}.
  \end{cases}
  \]
  We claim that $h(\mathbf{t})\le h(\mathbf{\hat t})$. To see this,
  first note that for every $i\in [d]$,
  $\prod_{j=1}^{w_i}t_{ij}=\prod_{j=1}^{w_i}\hat t_{ij}$, it is
  sufficient to prove that for every $i\in [d]$
  \[
  \sum_{j=1}^{w_i}\frac{1-t_{ij}}{\sqrt{t_{ij}}}\le
  \sum_{j=1}^{w_i}\frac{1-\hat t_{ij}}{\sqrt{\hat t_{ij}}}.
  \]
  This is a consequence of the Karamata's inequality by noticing that
  the function $\frac{1-e^{x}}{\sqrt{e^x}}$ is convex.

  We rename $\hat t_{i1}$ to $t_i$ and it is sufficient to upper bound
  \begin{equation}
    \label{eq:upmid}
    g(\mathbf{t},\mathbf{w})\defeq
    \sqrt{\frac{\lambda\prod_{i=1}^d\tuple{1-\tuple{\frac{\lambda}{1+\lambda}}^{w_i-1}t_i}}{1+\lambda\prod_{i=1}^d\tuple{1-\tuple{\frac{\lambda}{1+\lambda}}^{w_i-1}t_i}}}\cdot
    \sum_{i=1}^d\frac{w_i^c\tuple{\frac{\lambda}{1+\lambda}}^{w_i-1}t_i}{1-\tuple{\frac{\lambda}{1+\lambda}}^{w_i-1}t_i}\cdot
    \tuple{\frac{1-t_i}{\sqrt{t_i}}+\frac{(w_i-1)}{\sqrt{\lambda+\lambda^2}}}
  \end{equation}
  where $t_i\in\left[0,\frac{\lambda}{1+\lambda}\right]$ and
  $w_i\in\mathbb{Z}^+$ for every $i\in[d]$.

  The argument so far is similar to the proof in \cite{LiuL15}. In the
  following, we prove a monotonicity property of each $w_i$ and thus
  avoid the heavy numerical analysis in \cite{LiuL15} and allow us to
  obtain a tight result.

  For every $i\in[d]$, we let $z_i\defeq
  1-\tuple{\frac{\lambda}{1+\lambda}}^{w_i-1}t_i$ and thus
  equivalently
  $t_i=(1-z_i)\tuple{\frac{1+\lambda}{\lambda}}^{w_i-1}$. For every
  fixed $\mathbf{z}=(z_1,\dots,z_d)$, we can write \eqref{eq:upmid} as
  \begin{equation}
    \label{eq:upmid2}
    g_{\mathbf{z}}(\mathbf{w})=
    \sqrt{\frac{\lambda\prod_{i=1}^dz_i}{1+\lambda\prod_{i=1}^dz_i}}\sum_{i=1}^d\frac{1-z_i}{z_i}
    \tuple{\frac{1-t_i}{\sqrt{t_i}}+\frac{w_i-1}{\sqrt{\lambda+\lambda^2}}}w_i^c.
  \end{equation}
  We show that $g_{\mathbf{z}}(\mathbf{w})$ is monotonically
  decreasing with $w_i$ for every $i\in[d]$.

  Denote
  $T_i\defeq\frac{1-t_i}{\sqrt{t_i}}+\frac{(w_i-1)}{\sqrt{\lambda+\lambda^2}}$,
  then
  \begin{equation}
    \label{eq:upmid3}
    \frac{\partial g_{\mathbf{z}}(\mathbf{w})}{\partial w_i}=
    \sqrt{\frac{\lambda\prod_{i=1}^dz_i}{1+\lambda\prod_{i=1}^dz_i}}\cdot\frac{1-z_i}{z_i}
    \tuple{\frac{\partial T_i}{\partial w_i}w_i^c+cw_i^{c-1}T_i}.
  \end{equation}
  The partial derivative \eqref{eq:upmid3} is negative for a suitable
  choice of $c$:
  \begin{align*}
    &\quad\,\,\frac{1-z_i}{z_i}\cdot\tuple{\frac{\partial
        T_i}{\partial
        z_i}w_i^c+cw_i^{c-1}T_i} \\
    &=\frac{1-z_i}{z_i}\cdot\tuple{\tuple{-\frac{1}{2}t_i'(t_i^{-1/2}+t_i^{-3/2})+\frac{1}{\sqrt{\lambda+\lambda^2}}}w_i^c+cw_i^{c-1}
      \tuple{\frac{1-t_i}{\sqrt{t_i}}+\frac{(w_i-1)}{\sqrt{\lambda+\lambda^2}}}}\\
    &=\frac{1-z_i}{z_i}\cdot
    w_i^{c-1}\tuple{\tuple{-\frac{1}{2}\log\tuple{\frac{1+\lambda}{\lambda}}(t_i^{1/2}+t_i^{-1/2})+\frac{1}{\sqrt{\lambda+\lambda^2}}}w_i
      +c\tuple{t_i^{-1/2}-t_i^{1/2}+\frac{(w_i-1)}{\sqrt{\lambda+\lambda^2}}}}\\
    &=\frac{1-z_i}{z_i}\cdot w_i^{c-1} \tuple{
      \frac{(c+1)w_i-c}{\sqrt{\lambda+\lambda^2}}- \tuple{
        t_i^{1/2}\tuple{\frac{1}{2}w_i\log\tuple{\frac{1+\lambda}{\lambda}}+c}+
        t_i^{-1/2}\tuple{\frac{1}{2}w_i\log\tuple{\frac{1+\lambda}{\lambda}}-c}
      } }
  \end{align*}
  Denote
  \[
  p(t,w)\defeq \frac{(c+1)w-c}{\sqrt{\lambda+\lambda^2}}-
  \tuple{
    t^{1/2}\tuple{\frac{1}{2}w\log\tuple{\frac{1+\lambda}{\lambda}}+c}+
    t^{-1/2}\tuple{\frac{1}{2}w\log\tuple{\frac{1+\lambda}{\lambda}}-c}
  }
  \]
  Since $c\le\frac{\log(1+\lambda)-\log\lambda}{2+4\lambda}$, the term
  \[
  t^{1/2}\tuple{\frac{1}{2}w\log\tuple{\frac{1+\lambda}{\lambda}}+c}+
  t^{-1/2}\tuple{\frac{1}{2}w\log\tuple{\frac{1+\lambda}{\lambda}}-c}
  \]
  achieves its minimum at $t=\frac{\lambda}{1+\lambda}$. Thus
  \begin{align*}
    p(t,w) &\le p\tuple{\frac{\lambda}{1+\lambda},w}
    =\tuple{\frac{\lambda}{1+\lambda}}^{1/2}\tuple{\frac{c+1}{\lambda}-\frac{2\lambda+1}{2\lambda}\log\tuple{\frac{1+\lambda}{\lambda}}}w.
  \end{align*}
  
  Moreover, $c<\frac{2\lambda+1}{2}\log\tuple{\frac{1+\lambda}{\lambda}}-1$ implies that 
$\frac{c+1}{\lambda} < \frac{2\lambda+1}{2\lambda}\log\left(\frac{1+\lambda}{\lambda}\right)$
holds, which consequently leads to
$p\left(\frac{\lambda}{1+\lambda},1\right) < 0$.

In all, we choose a positive constant
$c<\min\left\{\frac{\log(1+\lambda)-\log\lambda}{2+4\lambda}, \frac{2\lambda+1}{2}\log\left(\frac{1+\lambda}{\lambda}\right)-1\right\}$, and this results in $p\left(\frac{\lambda}{1+\lambda},w\right)\le p\left(\frac{\lambda}{1+\lambda},1\right) < 0$.

  In light of the monotonicity of $w_i$'s, for every fixed
  $\mathbf{z}$, $g_{\mathbf{z}}(\mathbf{w})$ achieves its maximum when
  $\mathbf{w}=\mathbf{1}$. Thus
  \[
  \max_{\mathbf{t}\in\left[0,\frac{\lambda}{1+\lambda}\right]^d}g(\mathbf{t},\mathbf{w})
  =\max_{\substack{\mathbf{z}=(z_1,\dots,z_d)\\\forall i\in[d],
      z_i\in\left[1-\tuple{\frac{\lambda}{1+\lambda}}^{w_i-1},1\right]}}g_{\mathbf{z}}(\mathbf{w})
  \le\max_{\substack{\mathbf{z}=(z_1,\dots,z_d)\\\forall i\in[d],
      z_i\in\left[1-\tuple{\frac{\lambda}{1+\lambda}}^{w_i-1},1\right]}}g_{\mathbf{z}}(\mathbf{1})
  \le \max_{\mathbf{z}\in[0,1]^d}g_{\mathbf{z}}(\mathbf{1}).
  \]

  Actually, the case that all $w_i$'s are 1 corresponds to counting
  weighted independent sets on normal graphs and arguments to bound
  $g_{\mathbf{z}}(\mathbf{1})$ can be found in \cite{LLY13}. For the sake of completeness, we give a
  proof of $g_{\mathbf{z}}(\mathbf{1})\le \alpha<1$ (Lemma
  \ref{lem:ind-decay-tech}) in appendix.
\end{proof}

\begin{lemma}\label{lem:ind-decay-tech}
  Let $\Delta>1, $ be a constant. Assume
  $\lambda<\lambda_c=\frac{(\Delta-1)^{\Delta-1}}{(\Delta-2)^{\Delta}}$
  be a constant and $d<\Delta$. Then for some constant $\alpha<1$,
  $g_{\mathbf{z}}(\mathbf{1})\le\alpha<1$ where
  $\mathbf{z}=(z_1,\dots,z_d)\in[0,1]^d$.
\end{lemma}

We are now ready to prove the main lemma.

\begin{proof}[Proof of Lemma \ref{lem:ind-decay}]
    Let $\Phi(x)=\frac{1}{\sqrt{x(1+x)}}$ and
    $\phi(x)=\int\Phi(x)\,\mathrm{d}x=2\sinh^{-1}(\sqrt{x})$. We first
    apply induction on $\ell\defeq \max\set{0,L}$ to show that, if
    $d<\Delta$, then $\abs{\phi(R_v)-\phi(R(G,v,L))}\le 2\sqrt{\lambda}
    \alpha^L$ for some constant $\alpha<1$.
    
    If $\ell=0$, note that $R_v\in[0,\lambda]$, thus
    $\abs{\phi(R(G,v,L))-\phi(R_v)}\le 2\sqrt{\lambda}$. We now assume
    $L=\ell>0$ and the lemma holds for smaller $\ell$. For every
    $i\in[d]$ and $j\in [w_i]$, we denote $x_{ij}=R_{ij}$ and $\hat
    x_{ij}=R(G_{ij},v_{ij},L-\lfloor 1+c\log_{1/\alpha}
    w_i\rfloor)$. Let $\mathbf{x}=(x_{ij})_{i\in[d],j\in[w_i]}$,
    $\mathbf{\hat x}=(\hat x_{ij})_{i\in[d],j\in[w_i]}$. Let
    $f=\lambda\prod_{i=1}^d\tuple{1-\prod_{j=1}^{w_i}\frac{x_{ij}}{1-x_{ij}}}$
    , then it follows from Proposition \ref{prop:tech} that for some
    $\mathbf{\tilde x}=(\tilde x_{ij})_{i\in[d],j\in[w_i]}$ with each
    $\tilde x_{ij}\in[0,\lambda]$
    \begin{align*}
    \abs{\phi(R_v)-\phi(R(G,v,L))}
    &\le\sum_{i=1}^d\sum_{j=1}^{w_i}\frac{\Phi(f)}{\Phi(\tilde
        x_{ij})}\abs{\frac{\partial f(\mathbf{\tilde x})}{\partial
            x_{ij}}}\cdot\abs{\phi(x_{ij})-\phi(\hat x_{ij})} \\
    &\overset{(\spadesuit)}{\le}
    2\sqrt{\lambda}\sum_{i=1}^d\sum_{j=1}^{w_i}
    \frac{\Phi(f)}{\Phi(\tilde x_{ij})}\abs{\frac{\partial
            f(\mathbf{\tilde x})}{\partial
            x_{ij}}}\alpha^{L-\lfloor 1+c\log_{1/\alpha} w_i\rfloor}\\
    &\overset{(\heartsuit)}{\le} 2\sqrt{\lambda}\alpha^L.
    \end{align*}
    $(\spadesuit)$ follows from the induction hypothesis and
    $(\heartsuit)$ is due to Lemma \ref{lem:derivative}.
    
    The case that $d=\Delta$ can only happen at the root of our
    computational tree. Following the arguments in the proofs of
    \ref{lem:derivative}, \ref{lem:ind-decay-tech} and the bound in
    \eqref{eq:jensen}, it is easy to see that a universal constant
    upper bound for the error contraction exists, i.e.,
    \[
    \sum_{i=1}^{\Delta}\sum_{j=1}^{w_i}\frac{\Phi(f)}{\Phi(\tilde
        x_{ij})}\abs{\frac{\partial f(\mathbf{\tilde x})}{\partial
            x_{ij}}}<\max_{z\in[0,1]}\sqrt{\frac{\Delta^2(\Delta-1)^{\Delta-1}z^\Delta(1-z)}{(\Delta-2)^\Delta+(\Delta-1)^{\Delta-1}z^\Delta}}<
    3.
    \]
    Thus $\abs{\phi(R_v)-\phi(R(G,v,L))}\le 6\sqrt{\lambda} \alpha^L$
    for every $v$.
    
    Then the lemma follows from Proposition \ref{prop:tech}, since
    \begin{align*}
    \abs{R_v-R(G,v,L)}
    &=\frac{1}{\Phi(\tilde x)}\cdot\abs{\phi(R_v)-\phi(R(G,v,L))}&\mbox{for some $\tilde x\in[0,\lambda]$}\\
    &\le 6\lambda\sqrt{1+\lambda}\cdot\alpha^L
    \end{align*}
\end{proof}

\bibliographystyle{alpha} \bibliography{refs}

\appendix


\section{Proof of Theorem \ref{thm:ind-main}}

\begin{proof}
    The input of the FPTAS is an instance $G(V,\mathcal{E})$ and an
    accuracy parameter $0<\eps<1/2$.  Assume
    $V=\set{v_1,\dots,v_n}$. Note that $I=\varnothing$ is an independent
    set of $G$ with $w(I)=1$. Therefore
    \[
    Z(G)=1/\Pr[G]{I}=\tuple{\Pr[G]{\bigwedge_{i=1}^nv_i\not\in
            I}}^{-1}=\tuple{\prod_{i=1}^n\Pr[G]{v_i\not\in I\mid
            \bigwedge_{j=1}^{i-1}v_j\not\in I}}^{-1}.
    \]
    For every $1\le i\le n$, we define a graph $G_i(V_i,E_i)$:
    \begin{itemize}
        \item $G_1\defeq G$;
        \item For every $i\ge 2$, $G_i\defeq G_{i-1}-v_{i-1}-\mathcal{E}'$
        where $\mathcal{E}'\defeq \set{e\in \mathcal{E}_{i-1}\mid
            v_{i-1}\in e}$ consists of edges in $G_{i-1}$ incident to $v_i$.
    \end{itemize}
    It is straightforward to verify that $\Pr[G]{v_i\not\in
        I\mid\bigwedge_{j=1}^{i-1}v_j\not\in I}=\Pr[G_i]{v_i\not\in I}$
    for every $1\le i\le n$. Thus,
    \[
    Z(G)=\prod_{i=1}^n\tuple{\Pr[G_i]{v_i\not\in
            I}}^{-1}=\prod_{i=1}^n\tuple{1+R_i},
    \]
    where $R_i\defeq\frac{\Pr[G_i]{v_i\in I}}{\Pr[G_i]{v_i\not\in I}}$.
    Let $C$ and $\alpha$ be constants in Lemma \ref{lem:ind-decay}. We
    compute $R(G_i,v_i,L)$ with
    $L=\frac{\log\tuple{2Cn\eps^{-1}}}{\log\alpha^{-1}}$ for every $1\le
    i\le n$, then
    \[
    \abs{R_i-R(G_i,v_i,L)}\le\frac{\eps}{2n}.
    \]
    This implies
    \[
    1-\frac{\eps}{2n}\le \frac{1+R_i}{1+R(G,v_i,L)}\le
    1+\frac{\eps}{2n}.
    \]
    Let $\hat Z=\prod_{i=1}^n\tuple{1+R(G_i,v_i,L)}^{-1}$ be our estimate
    of the partition function, then it holds that
    \[
    e^{-\eps}\le \frac{Z(G)}{\hat Z}\le e^{\eps}.
    \]
    It remains to bound the running time of our algorithm. Let $T(L)$
    denote the maximum running time of computing $R(G,v,L)$ (over all
    choices of $d\le\Delta$ and arbitrary $w_i$). Then by the definition
    of $R(G,v,L)$, for every $L>0$,
    \[
    T(L)\le\sum_{i=1}^d\sum_{j=1}^{w_i}T(L-\lfloor
    1+c\log_{1/\alpha}w_i\rfloor) +O(n).
    \]
    It is easy to verify that
    $T(L)=n\Delta^{O(L)}=\tuple{\frac{n}{\eps}}^{O(\log\Delta)}$ for our choice of
    $L$. Thus our algorithm is an FPTAS for computing $Z(G)$.
\end{proof}

\section{Proof of Lemma \ref{lem:ind-recursion}}

\begin{proof}
    
    By the definition of $R_v$, we have
    \begin{align*}
    R_v =\frac{\Pr[G]{v\in I}}{\Pr[G]{v\not\in I}}
    =\lambda\cdot\frac{\Pr[G']{\bigwedge_{i=1}^dv_i\in
            I}}{\Pr[G']{\bigwedge_{i=1}^dv_i\not\in I}}
    =\lambda\cdot\prod_{i=1}^d\frac{\Pr[G']{v_i\in
            I\land\bigwedge_{j=1}^{i-1}v_j\not\in
            I\land\bigwedge_{j=i+1}^d v_j\in I}}{ \Pr[G']{v_i\not\in
            I\land\bigwedge_{j=1}^{i-1}v_j\not\in
            I\land\bigwedge_{j=i+1}^d v_j\in I} }
    \end{align*}
    For every $i\in[d]$, define $G_i\defeq G'-\set{v_k\mid i< k\le
        d}-\set{e_k\mid 1\le k<i}$, we have
    \[
    \frac{\Pr[G']{v_i\in I\land\bigwedge_{j=1}^{i-1}v_i\not\in
            I\land\bigwedge_{j=i+1}^d v_j\in I}}{ \Pr[G']{v_i\not\in
            I\land\bigwedge_{j=1}^{i-1}v_j\not\in I\land\bigwedge_{j=i+1}^d
            v_j\in I} } =\frac{\Pr[G']{v_i\in I\mid
            \bigwedge_{j=1}^{i-1}v_j\not\in I\land\bigwedge_{j=i+1}^d v_j\in
            I}}{ \Pr[G']{v_i\not\in I\mid \bigwedge_{j=1}^{i-1}v_j\not\in
            I\land\bigwedge_{j=i+1}^d v_j\in I}} = \frac{\Pr[G_i]{v_i\in
            I}}{\Pr[G_i]{v_i\not\in I}}.
    \]
    This is because fixing $v_j\in I$ is equivalent to removing $v_j$
    from the graph and fixing $v_j\not\in I$ is equivalent to removing
    all edges incident to $v_j$ from the graph.
    
    Since $e_i$ is the unique hyperedge in $G_i$ that contains $v_i$, we have
    \[
    \frac{\Pr[G_i]{v_i\in
            I}}{\Pr[G_i]{v_i\not\in I}}
    =1-\Pr[G_i-v_i-e_i]{\bigwedge_{j=1}^{w_i}v_{ij}\in I}
    =1-\prod_{j=1}^{w_i}\Pr[G_{ij}]{v_{ij}\in I}
    =1-\prod_{j=1}^{w_i}\frac{R_{ij}}{1+R_{ij}}.
    \]
\end{proof}

\section{Proof of Lemma \ref{lem:ind-decay-tech}}

\begin{proof}
    Let $\lambda_c'\defeq \frac{d^d}{(d-1)^{d+1}}$ be the uniqueness
    threshold for the $d$-ary tree. Then $\lambda<\lambda_c\le
    \lambda_c'$.
    
    Plugging $\mathbf{w}=\mathbf{1}$ into \eqref{eq:upmid2}, we have
    \[
    g_{\mathbf{z}}(\mathbf{1})
    =\sqrt{\frac{\lambda\prod_{i=1}^dz_i}{1+\lambda\prod_{i=1}^dz_i}}\cdot\sum_{i=1}^d\sqrt{1-z_i}.
    \]
    Let $z=\tuple{\prod_{i=1}^dz_i}^{\frac{1}{d}}$, it follows from
    Jensen's inequality that
    \begin{equation}
    \label{eq:jensen}
    g(\mathbf{t},\mathbf{1})\le d\sqrt{\frac{\lambda z^d(1-z)}{1+\lambda
            z^d}} < d\sqrt{\frac{\lambda_c'z^d(1-z)}{1+\lambda_c'z^d}}
    \end{equation}
    Recall that $f_{\lambda,d}(x)=\lambda\tuple{\frac{1}{1+x}}^d$. Let
    $\hat x$ be the positive fixed-point of $f_{\lambda_c',d}(x)$ and
    $\hat z=\frac{1}{1+\hat x}$. We show that
    $d\sqrt{\frac{\lambda_c'z^d(1-z)}{1+\lambda_c'z^d}}$ achieves its
    maximum when $z=\hat z$. The derivative of
    $\frac{\lambda_c'z^d(1-z)}{1+\lambda_c'z^d}$ with respect to $z$ is
    \[
    \tuple{\frac{\lambda_c'z^d(1-z)}{1+\lambda_c'z^d}}'=
    -\frac{\lambda_c'z^{d-1}}{\tuple{1+\lambda_c'z^d}^2}\tuple{z+\lambda_c'z^{d+1}-d(1-z)}.
    \]
    Since $\lambda_c'=\frac{d^d}{(d-1)^{d+1}}$, the above achieves
    maximum at $\tilde z=\frac{d-1}{d}$. If we let $\tilde
    x=\frac{1-\tilde z}{\tilde z}$, then it is easy to verify that
    $f_{\lambda_c',d}(\tilde x)=\tilde x$, which implies $\hat z=\tilde
    z$ because of the uniqueness of the positive fixed-point.
    
    Therefore, we have for some $\alpha<1$,
    \[
    g_{\mathbf{z}}(\mathbf{1})\le \alpha<d\sqrt{\frac{\lambda_c'\hat
            z^d(1-\hat z)}{1+\lambda_c'\hat z^d}}
    =\abs{f_{\lambda_c',d}'(\hat x)}=1.
    \]
\end{proof}

\section{FPTAS for Hypergraph Ising Model}\label{sec:ising-recursion}



\subsection{The algorithm and recursion for computing marginal probability}

In this section, we give the recursion function and design an algorithm to compute marginal probability in hypergraph Ising Model. Then we prove that the algorithm is indeed an FPTAS for any instance $G(V,\mathcal{E},(\bm\beta,\bm\gamma))$ if $1-\frac{2}{2e^{-1/2}\Delta+3}\le\bm\beta(e),\bm\gamma(e)\le 1$.

\paragraph{Graph Operations} 
We first fix some notations on graph
modification specific to our hypergraph two state spin model. Let
$G(V,\mathcal{E},(\bm\beta,\bm\gamma))$ be an instance. The first
groups of operations are about vertex removal and edge removal, which
is similar to the case of hardcore model:
\begin{itemize}
\item For every $v\in V$, we denote
  $G-v\defeq(V\setminus\set{v},\mathcal{E}',(\bm\beta',\bm\gamma'))$
  where $\mathcal{E}'\defeq\set{e-v\mid e\in \mathcal{E}}$ and
  $\bm\beta'(e-v)=\bm\beta(e)$, $\bm\gamma'(e-v)=\bm\gamma(e)$ for
  every $e\in\mathcal{E}$..
\item For every $e\in \mathcal{E}$, we denote
  $G-e\defeq(V,\mathcal{E}\setminus\set{e},(\bm\beta,\bm\gamma))$.
\item Let $x$ be a vertex or an edge and $y$ be a vertex or an edge,
  we denote $G-x-y\defeq (G-x)-y$.
\item Let $S=\set{v_1,\dots,v_k}\subseteq V$, we denote $ G-S\defeq
  G-v_1-v_2\cdots-v_k.  $
\item Let $\mathcal{F}=\set{e_1,\dots,e_k}\subseteq \mathcal{E}$, we
  denote $ G-\mathcal{F}\defeq G-e_1-e_2\cdots-e_k.$
\end{itemize}

The second group of operations is about \emph{pinning} the value of a
vertex, whose effect is to change $\bm\beta(e)$ or $\bm\gamma(e)$ for
edge $e$ incident to it.
\begin{itemize}
\item Let $v\in V$ be a vertex, we denote $G|_{v=0}\defeq
  (V\setminus\set{v},\mathcal{E}',(\bm\beta',\bm\gamma')_{e\in
    \mathcal{E}'})$ where $\mathcal{E}'\defeq \set{e-v\mid e\in
    \mathcal{E}}$, $\bm\beta'(e-v)\defeq \bm\beta(e)$ for every
  $e\in\mathcal{E}$ and
  \[
  \bm\gamma'(e-v)\defeq
  \begin{cases}
    \bm\gamma(e) & \mbox{ if } v\not\in e\\
    1 & \mbox{ otherwise.}
  \end{cases}
  \]
  This operation is to pin the value of $v$ to $0$.
\item Similarly, for a vertex $v\in V$, we denote $G|_{v=1}\defeq
  (V\setminus\set{v},\mathcal{E}',(\bm\beta',\bm\gamma'))$ where
  $\mathcal{E}'\defeq \set{e\setminus\set{v}\mid e\in \mathcal{E}}$,
  $\bm\gamma'(e-v)\defeq \bm\gamma(e)$ for every $e\in\mathcal{E}$ and
  \[
  \bm\beta'(e-v)\defeq
  \begin{cases}
    \bm\beta(e) &\mbox{if } v\not\in e\\
    1 &\mbox{otherwise.}
  \end{cases}
  \]
  This operation is to pin the value of $v$ to $1$.
\item Let $u,v\in V$ be two vertices and $i,j\in\set{0,1}$. We denote
  $G|_{u=i,v=j}\defeq \tuple{G|_{u=i}}|_{v=j}$ and this notation
  generalizes to more vertices.
\item Let $S=\set{v_1,\dots,v_k}\subseteq V$, we use $
  G|_{S=\mathbf{0}}$ (resp. $G|_{S=\mathbf{1}})$ to denote
  $G|_{v_1=0,v_2=0,\dots,v_k=0}$ and $G|_{v_1=1,v_2=1,\dots,v_k=1}$.
\end{itemize}


Let $G(V,\mathcal{E},(\bm\beta,\bm\gamma))$ be an instance and $v\in
V$ be an arbitrary vertex with degree $d$. Let
$E(v)=\set{e_1,\dots,e_d}$ be the set of edges incident to $v$. For
every $i\in[d]$, we assume $e_i=\set{v}\cup\set{v_{ij}\mid j\in[w_i]}$
consists of $w_i+1$ vertices where $w_i\ge 0$.

We define a new instance
$G'(V',\mathcal{E}',(\bm\beta_e',\bm\gamma_e'))$ which is obtained
from $G$ by replacing $v$ by $d$ copies of itself and each $e_i$
contains a distinct copy.  Formally,
$V'\defeq(V\setminus\set{v})\cup\set{v_1,\dots,v_d}$ and
$\mathcal{E}'\defeq\tuple{\mathcal{E}\setminus E(v)} \cup
\set{e_i-v+v_i\mid i\in[d]}$; $\bm\beta'(e)=\bm\beta(e)$,
$\bm\gamma'(e)=\bm\gamma(e)$ for every $e\in \mathcal{E}\setminus
E(v)$ and $\bm\beta'(e_i-v+v_i)=\bm\beta(e_i)$,
$\bm\gamma'(e_i-v+v_i)=\bm\gamma(e_i)$ for every $i\in[d]$.

For every $i\in[d]$ and $j\in[w_i]$, we define two smaller instances
$G^0_{ij}\tuple{V^0_{ij},\mathcal{E}^0_{ij},(\bm\beta_{ij}^0,\bm\gamma_{ij}^0)}$
and
$G^1_{ij}\tuple{V^1_{ij},\mathcal{E}^1_{ij},(\bm\beta_{ij}^1,\bm\gamma_{ij}^1)}$:

\begin{itemize}
\item $G^0_{ij}\defeq
  (((G'|_{V_{<i}=\mathbf{0}})|_{V_{>i}=\mathbf{1}})-e_i)|_{V^i_{<j}=\mathbf{0}}$;
\item $G^1_{ij}\defeq
  (((G'|_{V_{<i}=\mathbf{0}})|_{V_{>i}=\mathbf{1}})-e_i)|_{V^i_{<j}=\mathbf{1}}$,
\end{itemize}
where $V_{<i}\defeq \set{v_k\mid k<i}$, $V_{>i}\defeq \set{v_k\mid
  k>i}$, $V^i_{<j}\defeq \set{v_{ik}\mid k<j}$.

We now define the ratio of the probability on instances defined above.
\begin{itemize}
\item Let $R_v\defeq \frac{\Pr[G]{\sigma(v)=1}}{\Pr[G]{\sigma(v)=0}}$;
\item For every $i\in[d]$ and $j\in[w_i]$, let
  $R_{ij}^0\defeq\frac{\Pr[G_{ij}^0]{\sigma(v_{ij})=1}}{\Pr[G_{ij}^0]{\sigma(v_{ij})=0}}$,
  $R_{ij}^1\defeq\frac{\Pr[G_{ij}^1]{\sigma(v_{ij})=1}}{\Pr[G_{ij}^1]{\sigma(v_{ij})=0}}$.
\end{itemize}
It follows from our definition that for every $i\in[d]$,
$R_{i1}^0=R_{i1}^1$.

We can compute $R_v$ by the following recursion:
\begin{lemma}\label{lem:ising-recursion}
  \begin{equation}
    R_v=\lambda
    \cdot\prod_{i=1}^{d}\frac{1-(1-\bm\gamma(e_i))\frac{R_{i1}^0}{1+R_{i1}^0}\prod_{j=2}^{w_i}\frac{R_{ij}^1}{1+R_{ij}^1}}{1-(1-\bm\beta(e_i))\frac{1}{1+R_{i1}^0}\prod_{j=2}^{w_i}\frac{1}{1+R_{ij}^0}}.  \label{eq:ising-recursion}
  \end{equation}

\end{lemma}
\begin{proof}
  By the definition of $R_v$, we have
  \[
  R_v=\frac{\Pr[G]{\sigma(v)=1}}{\Pr[G]{\sigma(v)=0}}
  =\lambda\cdot\frac{\Pr[G']{\bigwedge_{i=1}^d\sigma(v_i)=1}}{\Pr[G']{\bigwedge_{i=1}^d\sigma(v_i)=0}}
  =\lambda\cdot\prod_{i=1}^d\frac{\Pr[G']{\sigma(v_i)=1\land\bigwedge_{j=1}^{i-1}\sigma(v_j)=0\land\bigwedge_{j=i+1}^{d}\sigma(v_j)=1}}{\Pr[G']{\sigma(v_i)=0\land\bigwedge_{j=1}^{i-1}\sigma(v_j)=0\land\bigwedge_{j=i+1}^{d}\sigma(v_j)=1}}.
  \]
  For every $i\in[d]$, define $G_i\defeq
  (G'|_{V_{<i}=\mathbf{0}})|_{V_{>i}=\mathbf{1}}$, then
  \begin{align*}
    \frac{\Pr[G']{\sigma(v_i)=1\land\bigwedge_{j=1}^{i-1}\sigma(v_j)=0\land\bigwedge_{j=i+1}^{d}\sigma(v_j)=1}}{\Pr[G']{\sigma(v_i)=0\land\bigwedge_{j=1}^{i-1}\sigma(v_j)=0\land\bigwedge_{j=i+1}^{d}\sigma(v_j)=1}}
    &=
    \frac{\Pr[G']{\sigma(v_i)=1\mid\bigwedge_{j=1}^{i-1}\sigma(v_j)=0\land\bigwedge_{j=i+1}^{d}\sigma(v_j)=1}}{\Pr[G']{\sigma(v_i)=0\mid\bigwedge_{j=1}^{i-1}\sigma(v_j)=0\land\bigwedge_{j=i+1}^{d}\sigma(v_j)=1}}\\
    &=\frac{\Pr[G_i]{\sigma(v_i)=1}}{\Pr[G_i]{\sigma(v_i)=0}}.
  \end{align*}
  Since $e_i$ is the unique edge in $G_i$ that contains $v_i$, we have
  \begin{align*}
    \frac{\Pr[G_i]{\sigma(v_i)=1}}{\Pr[G_i]{\sigma(v_i)=0}}
    &=\frac{\bm\gamma(e_i)\Pr[G_i-v_i-e_i]{\bigwedge_{j=1}^{w_i}\sigma(v_{ij})=1}+\tuple{1-\Pr[G_i-v_i-e_i]{\bigwedge_{j=1}^{w_i}\sigma(v_{ij})=1}}}{\bm\beta(e_i)
      \Pr[G_i-v_i-e_i]{\bigwedge_{j=1}^{w_i}\sigma(v_{ij})=0}+\tuple{1-\Pr[G_i-v_i-e_i]{\bigwedge_{j=1}^{w_i}\sigma(v_{ij})=0}}}\\
    &=\frac{1-(1-\bm\gamma(e_i))\Pr[G_i-v_i-e_i]{\sigma(v_{i1})=1}\cdot\prod_{j=2}^{w_i}\Pr[G_i-v_i-e_i]{\sigma(v_{ij})=1\mid \bigwedge_{k=1}^{j-1}\sigma(v_k)=1}}{1-(1-\bm\beta(e_i))\Pr[G_i-v_i-e_i]{\sigma(v_{i1})=0}\cdot\prod_{j=2}^{w_i}\Pr[G_i-v_i-e_i]{\sigma(v_{ij})=0\mid \bigwedge_{k=1}^{j-1}\sigma(v_k)=0}}\\
    &=\frac{1-(1-\bm\gamma(e_i))\Pr[G^1_{i1}]{\sigma(v_{i1})=1}\cdot\prod_{j=2}^{w_i}\Pr[G^1_{ij}]{\sigma(v_{ij})=1}}{1-(1-\bm\beta(e_i))\Pr[G^0_{i1}]{\sigma(v_{i1})=0}\cdot\prod_{j=2}^{w_i}\Pr[G^0_{ij}]{\sigma(v_{ij})=0}}\\
    &=\frac{1-(1-\bm\gamma(e_i))\frac{R^0_{i1}}{1+R^0_{i1}}\prod_{j=2}^{w_i}\frac{R^1_{ij}}{1+R^1_{ij}}}{1-(1-\bm\beta(e_i))\frac{1}{1+R^0_{i1}}\prod_{j=2}^{w_i}\frac{1}{1+R^0_{ij}}}.
  \end{align*}
  The last equality is due to the fact that $R^0_{i1}=R^1_{i1}$.

\end{proof}

The algorithm description is similar to Section
\ref{sec:indset-algorithm}.  Let
$G(V,\mathcal{E},(\bm\beta,\bm\gamma))$ be an instance of our two
state spin model with maximum degree $\Delta$, and $v\in V$ be an
arbitrary vertex with degree $d$. Let $E(v)=\set{e_1,\dots,e_d}$
denote the set of edges incident to $v$. Define $G^0_{ij}$,
$G^1_{ij}$, $R_v$, $R^0_{ij}$, $R^1_{ij}$ as in Section
\ref{sec:ising-recursion}. Then the recursion
\eqref{eq:ising-recursion} gives a way to compute the marginal
probability $\Pr[G]{\sigma(v) =1}$ \emph{exactly}. However, an exact
evaluation of the recursion requires a computation tree with
exponential size. Thus we introduce the following truncated version of
the recursion, with respect to constants $c>0$ and $0<\alpha<1$.  Let
\begin{align*}
  f_G(\mathbf{r}) & = \lambda \cdot \prod_{i=1}^{d}
  \frac{1-(1-\bm\gamma(e_i))\frac{r_{i1}^0}{1+r_{i1}^0}\prod_{j=2}^{w_i}
    \frac{r_{ij}^1}{1+r_{ij}^1}}{1-(1-\bm\beta(e_i))\frac{1}{1+r_{i1}^0}\prod_{j=2}^{w_i}\frac{1}{1+r_{ij}^0}},
\end{align*}
where $\mathbf{r} \defeq ((r^0_{ij})_{1 \le j \le w_i}, (r^1_{ij})_{2
  \le j \le w_i})_{i \in [d]}$ with
\begin{align*}
  r^0_{ij} &= R(G^0_{ij},v_{ij},L-\lfloor
  1+c\log_{1/\alpha} w_i\rfloor)\\
  r^1_{ij} &= R(G^1_{ij},v_{ij},L-\lfloor 1+c\log_{1/\alpha}
  w_i\rfloor).
\end{align*}

We can describe our truncated recursion as follows:
\[
R(G,v,L) =
\begin{cases}
  f_G(\mathbf{r})  & \mbox{if } L > 0\\
  \lambda & \mbox{otherwise.}
\end{cases}
\]

The recursion can be directly used to compute $R(G,v,L)$ for any given
$L$ and it induces a truncated computation tree (with height $L$ in
some special
metric). 

We claim that $R(G,v,L)$ is a good estimate of $R_v$ for a suitable
choice of $c$ and $\alpha$, for those $(\bm\beta,\bm\gamma,\Delta)$
satisfying that $1 - \frac{2}{2e^{-1/2}\Delta + 3}
\le\bm\beta(e),\bm\gamma(e)\le 1$ for all $e\in\mathcal{E}$.

\begin{lemma}\label{lem:ising-derivative}
  Let $G(V,\mathcal{E}, (\bm\beta, \bm\gamma))$ be an instance of our
  generalized two state spin system model with maximum degree
  $\Delta\ge 2$. Let $v\in V$ be a vertex and assume
  $1 - \frac{2}{2e^{-1/2}\Delta + 3}\le\bm\beta(e),\bm\gamma(e)\le 1$
  for all $e\in\mathcal{E}$. There exist constants $C>0$ and
  $0<\alpha<1$ such that
  \[
  \abs{R(G,v,L)-R_v}\le C\cdot\alpha^{\max\{0,L\}}
  \]
  for every $L$.
\end{lemma}

The Lemma \ref{lem:ising-derivative} will be proved in the next
section. The proof of Theorem \ref{thm:twostate-main} with this lemma
is almost identical to the proof of Theorem \ref{thm:ind-main} and
we omit it here.

\subsection{Correlation decay}


Again, the key is to bound the amortized decay rate as in the
following lemma.

\begin{lemma}
  \label{lem:ising-decay}
  Let $\Delta \ge 2$ be a constant integer. Suppose $d \le \Delta$ is
  a constant integer, $\lambda \in (0, 1), \beta_c =1 -
  \frac{2}{2e^{-1/2}\Delta+3}$ are constant real numbers and
  $\gamma_1, \gamma_2, \ldots \gamma_d, \beta_1, \beta_2, \ldots,
  \beta_d \in [\beta_c, 1]$. Let $w_1,\dots,w_d>0$ be integers and %
  \[f(\mathbf{r})=\lambda\prod_{i=1}^{d} \frac{1-(1 -
    \gamma_i)\frac{r_{i1}^0}{1+r_{i1}^0}\prod_{j=2}^{w_i}\frac{r_{ij}^1}{1+r_{ij}^1}}{1-(1
    - \beta_i)
    \frac{1}{1+r_{i1}^0}\prod_{j=2}^{w_i}\frac{1}{1+r_{ij}^0}}\] be a
  $\tuple{\sum_{i=1}^dw_i}$-ary function. Let $\Phi(x)=\frac{1}{x}$.
  There exist constants $\alpha<1$ and $c > 0$ depending on $\gamma$
  and $\Delta$ (but not depending on $w_i$) such that

  \[
  \sum_{i=1}^{d} w_i^c \tuple{\sum_{j=1}^{w_i}
    \frac{\Phi(f)}{\Phi(r^0_{ij})}\abs{\frac{\partial
        f(\mathbf{r})}{\partial r^0_{ij}}} + \sum_{j = 2}^{w_i}
    \frac{\Phi(f)}{\Phi(r^1_{ij})}\abs{\frac{\partial
        f(\mathbf{r})}{\partial r^1_{ij}}}} \le \alpha<1
  \]
  for every ${\bf r}=(r_{ij})_{i\in[d],j\in[w_i]}$ where each
  $r_{ij}\in [\lambda\beta_c^{\Delta}, \lambda\beta_c^{-\Delta}]$.

\end{lemma}

\begin{proof}
  We first let $x_i \defeq \frac{r^0_{i1}}{1 + r^0_{i1}}, y_{ij}
  \defeq \frac{r^1_{ij}}{1 + r^1_{ij}}, z_{ij} \defeq
  \frac{r^0_{ij}}{1 + r^0_{ij}}$ for all $1 \le i \le d$ and $2 \le j
  \le w_i$.  The it holds that $x_i, y_{ij}, z_{ij} \in
  \left[\frac{\lambda\beta_c^ \Delta}{1 + \lambda\beta_c^ \Delta},
    \frac{\lambda}{\lambda + \beta_c^ \Delta}\right]$.  To simplify the
  notation, we denote
  \begin{align*}
    A_i&\defeq 1-(1-\gamma_i)x_i\prod_{j=2}^{w_i}y_{ij}\\
    B_i&\defeq 1-(1-\beta_i)(1-x_i)\prod_{j=2}^{w_i}(1-z_{ij})
  \end{align*}
  where $A_i \in [\gamma_i, 1] \subseteq [\beta_c, 1], B_i \in
  [\beta_i, 1] \subseteq [\beta_c, 1]$.

  Then we can write $f$ as
  \[f = \lambda\prod_{i = 1}^d \frac{A_i}{B_i}.\] We can directly
  compute the partial derivatives of $f$, which yields the following
  for all $i\in[d]$ and $2\le j\le w_i$:
  \begin{align*}
    \frac{\partial f}{\partial r^0_{i1}} &=
    \tuple{\lambda\prod_{\substack{k\in[d]\\k\ne
          i}}\frac{A_k}{B_k}}\cdot
    \frac{(1-x_i)\tuple{(1-A_i)(1-B_i)-(1-x_i)(1-A_i)-x_i(1-B_i)}}{x_iB_i^2}\\
    & = f\cdot\frac{(1-x_i)\tuple{(1-A_i)(1-B_i)-(1-x_i)(1-A_i)-x_i(1-B_i)}}{x_iA_iB_i};\\
    \frac{\partial f}{\partial r^1_{ij}} &=
    \tuple{\lambda\prod_{\substack{k\in[d]\\k\ne
          i}}\frac{A_k}{B_k}}\cdot
    \frac{-(1-y_{ij})^2(1-A_i)}{y_{ij}B_i}= f\cdot\frac{-(1-y_{ij})^2(1-A_i)}{y_{ij}A_i};\\
    \frac{\partial f}{\partial r^0_{ij}} &=
    \tuple{\lambda\prod_{\substack{k\in[d]\\k\ne
          i}}\frac{A_k}{B_k}}\cdot \frac{-A_i(1-B_i)(1-z_{ij})}{B_i^2}
    = f \cdot\frac{-(1-B_i)(1-z_{ij})}{B_i}.
  \end{align*}
  Thus,
  \begin{align*}
    &\quad\;\sum_{i=1}^d w_i^c \tuple{\sum_{j=1}^{w_i}
      \frac{\Phi(f)}{\Phi(r^0_{ij})}\abs{\frac{\partial
          f(\mathbf{r})}{\partial r^0_{ij}}} + \sum_{j = 2}^{w_i} \frac{\Phi(f)}{\Phi(r^1_{ij})}\abs{\frac{\partial f(\mathbf{r})}{\partial r^1_{ij}}}}\\
    &= \frac{1}{f(\mathbf{r})} \sum_{i=1}^d w_i^c
    \tuple{\frac{x_i}{1-x_i}\abs{\frac{\partial
          f(\mathbf{r})}{\partial r^0_{i1}}} + \sum_{j =
        2}^{w_i}\frac{y_i}{1-y_i}\abs{\frac{\partial
          f(\mathbf{r})}{\partial r^1_{ij}}}
      + \sum_{j = 2}^{w_i}\frac{z_i}{1-z_i}\abs{\frac{\partial f(\mathbf{r})}{\partial r^0_{ij}}}}\\
    &= \sum_{i = 1}^d
    w_i^c\left(\frac{-(1-A_i)(1-B_i)+(1-x_i)(1-A_i)+x_i(1-B_i)}{A_iB_i}+\frac{1-A_i}{A_i}\sum_{j
        = 2}^{w_i} (1-y_{ij}) + \frac{1-B_i}{B_i}\sum_{j = 2}^{w_i}
      z_{ij}\right)
  \end{align*}
  Let $y_i\defeq \tuple{\prod_{j = 2}^{w_i} y_{ij}}^{1/(w_i-1)}, z_i
  \defeq 1 - \tuple{\prod_{j = 2}^{w_i} 1 - z_{ij}}^{1/(w_i-1)}$ for
  every $i\in[d]$, then it holds that
  \begin{align*}
    A_i &= 1 - (1 - \gamma_i)x_i y_i^{w_i - 1},\\ %
    B_i &=  1 - (1 - \beta_i)(1-x_i) (1-z_i)^{w_i - 1},
  \end{align*}
  \begin{align*}
    \sum_{j =
      2}^{w_i} (1 - y_{ij}) &\le (w_i - 1)(1 - y_i)\\
    \sum_{j = 2}^{w_i} z_{ij} &\le (w_i - 1)z_i.
  \end{align*}
  Since $A_i, B_i \in [\beta_c, 1]$ for every $i\in[d]$, we have
  \begin{align*}
    &\quad\; \sum_{i = 1}^d
    w_i^c\tuple{\frac{-(1-A_i)(1-B_i)+(1-x_i)(1-A_i)+x_i(1-B_i)}{A_iB_i}+\frac{1-A_i}{A_i}\sum_{j
        = 2}^{w_i} (1-y_{ij}) + \frac{1-B_i}{B_i}\sum_{j = 2}^{w_i}z_{ij}}\\
    &\leq  \sum_{i = 1}^d w_i^c\tuple{\frac{(1-x_i)(1-A_i)+x_i(1-B_i)-(1-A_i)(1-B_i)}{A_iB_i} + (w_i-1)\tuple{\frac{(1 - A_i)(1 - y_i)}{A_i} + \frac{(1 - B_i)z_i}{B_i}}}\\
    &\leq  \sum_{i = 1}^d w_i^c \tuple{\frac{(1-x_i)(1-A_i)+x_i(1-B_i)}{A_iB_i}+(1 - \beta_c)(w_i-1)\tuple{\frac{x_i y_i^{w_i - 1}(1-y_i)}{A_i}+\frac{(1-x_i)(1-z_i)^{w_i - 1}z_i}{B_i}}}\\
    &\leq \frac{1-\beta_c}{\beta_c^2} \sum_{i = 1}^d w_i^c \tuple{(1 -
      x_i)x_i(y_i^{w_i - 1} + (1-z_i)^{w_i - 1}) + \beta_c (w_i -
      1)(x_iy_i^{w_i - 1}(1-y_i)+(1-x_i)(1-z_i)^{w_i -
        1}z_i)}
  \end{align*}

  We now assume Lemma \ref{lem:ising-bound-item}, thus there exist
  constants $\alpha<1,c > 0$ depending on $\lambda$ and $\Delta$ such
  that
  \[
  w_i^c\tuple{(1 - x_i)x_i(y_i^{w_i - 1} + (1-z_i)^{w_i - 1}) +
    \beta_c (w_i - 1)(x_iy_i^{w_i - 1}(1-y_i)+(1-x_i)(1-z_i)^{w_i -
      1}z_i)} \le \alpha\beta_c e^{\frac{1}{2\beta_c}-1}
  \]
  for every $i\in[d]$.
  Thus,
  \begin{align*}
    &\quad\; \sum_{i = 1}^d w_i^c\tuple{\frac{-(1-A_i)(1-B_i)+(1-x_i)(1-A_i)+x_i(1-B_i)}{A_iB_i}+\frac{1-A_i}{A_i}\sum_{j = 2}^{w_i} (1-y_{ij}) + \frac{1-B_i}{B_i}\sum_{j = 2}^{w_i}z_{ij}}\\
    &\le  \frac{1-\beta_c}{\beta_c} \cdot d \cdot \alpha \cdot e^{\beta_c^{-1}/2-1}\\
    &\le \frac{2\alpha d}{2e^{-1/2}d+3}\cdot \frac{2e^{-1/2}d+3}{2e^{-1/2}d+1}\cdot e^{(1 -2/(2e^{-1/2}d+3))^{-1} /2-1}\\
    &= \alpha \cdot \frac{2d}{2e^{-1/2}d+1}\cdot e^{(1 -
      2/(2e^{-1/2}d+3))^{-1}/2-1}\\
    &\le \alpha.
  \end{align*}
  The last inequality is due to the fact that $h(d)\defeq
  \frac{2d}{2e^{-1/2}d+1}\cdot e^{(1 - 2/(2e^{-1/2}d+3))^{-1}/2-1}$ is
  an increasing function on $d$ and $\lim_{d\to\infty}h(d)=1$.

  Combining all above, there exist constants $\alpha < 1$ and $c > 0$
  such that for all $w_1, w_2, \ldots w_d > 0$ and every ${\bf
    r}=(r_{ij})_{i\in[d],j\in[w_i]}$ with each $r_{ij}\in
  [\lambda\beta_c^{\Delta}, \lambda\beta_c^{-\Delta}]$, it holds
  \[
  \sum_{i=1}^{d} w_i^c \tuple{\sum_{j=1}^{w_i}
    \frac{\Phi(f)}{\Phi(r^0_{ij})}\abs{\frac{\partial
        f(\mathbf{r})}{\partial r^0_{ij}}} + \sum_{j = 2}^{w_i}
    \frac{\Phi(f)}{\Phi(r^1_{ij})}\abs{\frac{\partial
        f(\mathbf{r})}{\partial r^1_{ij}}}} \le \alpha<1
  \]
\end{proof}

It remains to prove Lemma \ref{lem:ising-bound-item}:

\begin{lemma}\label{lem:ising-bound-item}
  For all $\beta < 1$ and $0 < \delta < 1/2$, there exist constants
  $\alpha < 1$ and $c > 0$ such that
  \[w^c\tuple{(1 - x)x(y^{w - 1} + (1-z)^{w - 1}) + \beta (w - 1)xy^{w
      - 1}(1-y)+\beta (w - 1)(1-x)(1-z)^{w - 1}z} \le \alpha \beta
  e^{\frac{1}{2\beta}-1}\] for all $w \in \mathbb{N}$ and all $x, y, z
  \in [\delta, 1 - \delta]$.
\end{lemma}
\begin{proof}
  For a fixed $w$, let
  \[
  V_w(x, y, z)\defeq (1 - x)x(y^{w - 1} + (1-z)^{w - 1}) + \beta (w -
  1)xy^{w - 1}(1-y)+\beta (w - 1)(1-x)(1-z)^{w - 1}z.\] Then $V_w(x,
  y, z)$ achieves maximum value when
  \begin{align*}
    y & = y_c \defeq \frac{\beta (w - 1) + 1 - x}{\beta w} = 1 - \frac{1 - (1 - x)/\beta}{w},\\
    z & = z_c \defeq 1 - \frac{\beta (w - 1) + x}{\beta w} = \frac{1 -
      x/\beta}{w}.
  \end{align*}
  If $y < 1 - \frac{1 - (1 - x)/\beta}{w}$, $V_w(x, y, z)$ is an
  increasing function on $y$. If $z > \frac{1 - x/\beta}{w}$, $V_w(x,
  y, z)$ is a decreasing function on $z$.

  Let $\hat w \defeq \max\set{\left\lceil\frac{1}{\delta}\right\rceil,
    \left\lceil-\frac{1}{\ln(1-\delta)}\right\rceil}+1$. If $w>\hat
  w$, the monotonicity of $V_w(x,y,z)$ implies that the function
  achieves it maximum when the $y$ and $z$ are at the boundary, i.e.,
  $y=1-z=1-\delta$, for every fixed $x$.

  We now analyze the case $w\le \hat w$ and the case $w>\hat w$
  separately. Define $\alpha \defeq
  \tuple{\frac{\tuple{1+\frac{1/(2\beta) - 1}{\hat w}}^{\hat
        w}}{e^{1/(2\beta)-1}}}^{1/2}$. It is easy to verify that $\alpha < 1$ since $(1+x/n)^n < e^x$ for every $x>0$.
  \begin{itemize}
  \item (If $w \le \hat w$.) We plug $y=y_c$ and $z=z_c$ into
    $V_w(x,y,z)$:
    \[V_w(x, y, z) \le (1 - x)x(y_c^{w - 1} + (1-z_c)^{w - 1}) + \beta
    (w - 1)xy_c^{w - 1}(1-y_c)+\beta (w - 1)(1 - x) (1-z_c)^{w -
      1}z_c.\] The derivative of the right hand side with respect to
    $x$ is
    \[\frac{(\beta (w - 1)+1-x)^{w - 1}(1 + \beta (w - 1) - (w + 1)x)
      - (\beta (w - 1) + x)^{w - 1}(1 + \beta (w - 1) - (w +
      1)(1-x))}{(\beta w)^w}.\]
    It is easy to see that the function above is zero when $x=1/2$ and
    positive for smaller $x$, negative for larger $x$.

    Thus, when $x = 1/2, y = y_c, z = z_c$, $V_w(x, y, z)$ achieves
    its maximum value
    \[V_w\tuple{\frac{1}{2}, 1 - \frac{1 - \frac{1}{2\beta}}{w},
      \frac{1 - \frac{1}{2\beta}}{w}} =
    \beta\left(1-\frac{1-\frac{1}{2\beta}}{w}\right)^{w} \le
    \beta\left(1-\frac{1-\frac{1}{2\beta}}{\hat w}\right)^{\hat w} <
    \beta e^{\frac{1}{2\beta}-1}.\]%
    Let $c_1 \defeq \log_{\hat w} \frac{\alpha
      e^{\frac{1}{2\beta}-1}}{\tuple{1-\frac{1-\frac{1}{2\beta}}{\hat
          w}}^{\hat w}} = -\log_{\hat w} \alpha > 0$, then for all $w
    \le \hat w$ and $c' \in (0, c_1]$,
    \[w^{c'}V_w(x, y, z) \le w^{c_1}V_w(x, y, z) \le \hat
    w^{c_1}\beta\left(1-\frac{1-\frac{1}{2\beta}}{\hat w}\right)^{\hat
      w} = \frac{\beta}{\alpha}\tuple{1-\frac{1-\frac{1}{2\beta}}{\hat
        w}}^{\hat w} = \alpha\beta e^{\frac{1}{2\beta}-1}.\]
  \item (If $w > \hat w$.) Let $c_2\defeq -\hat w \log(1 - \delta) -
    1> 0$. Then
    \[
    V_w(x, y, z) \le V_w(x, 1 - \delta, \delta) = 2x(1-x)(1-\delta)^{w
      - 1} + \beta(w - 1)\delta(1-\delta)^{w -
      1}\]
    The above achieves its maximum when $x=1/2$, i.e., $V_w(x,y,z)\le
    V_w\tuple{\frac{1}{2},1-\delta,\delta}$. Therefore
    \[
    w^{c_2}\cdot V_w(x,y,z)\le w^{c_2}\cdot
    V_w\tuple{\frac{1}{2},1-\delta,\delta}=w^{c_2}\cdot \frac{1+2\beta
      (w - 1)\delta}{2}(1-\delta)^{w - 1}.
    \]

    We now prove that for all $c' \in (0, c_2]$, $g(w) \defeq
    w^{c'}\frac{1+2\beta(w-1)\delta}{2}(1-\delta)^{w-1}$ is decreasing
    on $w$ when $w > \hat w$.  Since $c' \le c_2 = -\hat w \log(1 -
    \delta) - 1 \le -(w - 1) \log(1 - \delta) - 1$ for all $w > \hat
    w$,
    \begin{align*}
      \frac{\partial g}{\partial w} &= \frac{w^{c'-1}(1-\delta)^{w-1}}{2} \tuple{c'+w\log(1-\delta)}+ \beta\delta w^{c'-1}(1-\delta)^{w-1}\tuple{w+(w-1)(c'+w\log(1-\delta))}\\
      &\le \frac{w^{c'-1}(1-\delta)^{w-1}}{2}\tuple{\log(1-\delta) - 1} + \beta\delta w^{c'-1}(1-\delta)^{w-1}\tuple{w+(w-1)(\log(1-\delta)-1)}\\
      &\le -\frac{w^{c'-1}(1-\delta)^{w-1}}{2} + \beta\delta w^{c'-1}(1-\delta)^{w-1}(1 + (w-1)\log(1-\delta))\\
      &\le -\frac{w^{c'-1}(1-\delta)^{w-1}}{2} - \beta\delta w^{c'-1}(1-\delta)^{w-1}c')\\
      &< 0
    \end{align*}
    Thus, $g(w)$ is decreasing on $w$ when $w > \hat w$.
    In light of this, $w^{c'}\cdot V_w(x,y,z)$ can be upper bounded by
    $\hat w^{c'}\cdot V_{\hat w}\tuple{\frac{1}{2},1-\delta,\delta}$
    for all $c' \in (0, c_2]$.

    In all, we have
    \begin{align*}
      w^{c'}\cdot V_w(x, y, z) & \le w^{c'}\cdot V_w\tuple{\frac{1}{2}, 1 - \delta, \delta} \\
      & \le \hat w^{c'}\cdot V_{\hat w}\tuple{\frac{1}{2}, 1 - \delta, \delta}\\
      & \le \hat w^{c_1}\cdot V_{\hat w}\tuple{\frac{1}{2}, 1 - \frac{1 - \frac{1}{2\beta}}{\hat w}, \frac{1 - \frac{1}{2\beta}}{\hat w}} \hat w^{c' - c_1} \\
      & = \alpha \beta e^{\frac{1}{2\beta}-1} \hat w^{c'-c_1}
    \end{align*}
    for all $w > \hat w$ and $c' \in (0, c_2]$.
  \end{itemize}
  Let $c = \min\{c_1, c_2\}$, then we have $w^cV_w(x,y,z) \le
  w^{c_1}V_w(x, y, z) \le \alpha\beta e^{\frac{1}{2\beta}-1}$ if $w
  \le \hat w$ and $w^cV_w(x,y,z) \le \hat w^{c_1}V_{\hat w}(x, y,
  z)\hat w^{c'-c_1} \le \alpha\beta e^{\frac{1}{2\beta}-1}$ if $w >
  \hat w$, i.e.,
  \[w^c\tuple{(1 - x)x(y^{w - 1} + (1-z)^{w - 1}) + \beta (w - 1)xy^{w
      - 1}(1-y)+\beta (w - 1)(1-x)(1-z)^{w - 1}z} \le \alpha \beta
  e^{\frac{1}{2\beta}-1}\] for all $w \in \mathbb{N}$ and all $x, y, z
  \in [\delta, 1 - \delta]$.

\end{proof}

We are now ready to prove Lemma \ref{lem:ising-derivative}.

\begin{proof}[Proof of Lemma \ref{lem:ising-derivative}]
  Let $\Phi(x)=\frac{1}{x}$ and $\phi(x)=\int\Phi(x)\,\mathrm{d}x=\ln
  x$. We apply induction on $\ell\defeq \max\set{0,L}$ to show that,
  if $d\le\Delta$, then $\abs{\phi(f_v)-\phi(f(G,v,L))}\le
  4\sqrt{e}\alpha^L$ for some constant $\alpha<1$.

  If $\ell=0$, note that $f_v\in[\lambda\beta_c^ \Delta,
  \lambda\beta_c^{-\Delta}]$, thus $\abs{\phi(f(G,v,L))-\phi(f_v)} \le
  -2\Delta\ln \beta_c$. Since $\abs{\ln(1-x)} \le 2x$ for all
  $x\in\tuple{0,\frac{1}{2}}$,
  \[\abs{\phi(f(G,v,L))-\phi(f_v)} \le -2\Delta\ln \beta_c \le
  4\Delta\frac{2}{\frac{2}{\sqrt{e}}\Delta+3} \le 4\sqrt{e}.\]

  We now assume $L=\ell>0$ and the lemma holds for smaller $\ell$. For
  every $i\in[d]$ and $j\in [w_i]$, we denote $r^0_{ij}=R^0_{ij},
  r^1_{ij}= R^1_{ij}$ and $\hat r^0_{ij} = R(G^0_{ij},v_{ij},L-\lfloor
  1+c\log_{1/\alpha} w_i\rfloor), \hat r^1_{ij} =
  R(G^1_{ij},v_{ij},L-\lfloor 1+c\log_{1/\alpha} w_i\rfloor)$. Let
  $\mathbf{r}=((r^0_{ij})_{1\le j \le w_i}, (r^1_{ij})_{2\le j \le
    w_i})_{i \in [d]}$, $\mathbf{\hat r}=((\hat r^0_{ij})_{1\le j \le
    w_i}, (\hat r^1_{ij})_{2\le j \le w_i})_{i \in [d]}$. Let
  $f(\mathbf{r})=\lambda\prod_{i=1}^{d} \frac{1-(1 -
    \gamma_i)\frac{r_{i1}^0}{1+r_{i1}^0}\prod_{j=2}^{w_i}\frac{r_{ij}^1}{1+r_{ij}^1}}{1-(1
    - \beta_i)
    \frac{1}{1+r_{i1}^0}\prod_{j=2}^{w_i}\frac{1}{1+r_{ij}^0}}$, then
  it follows from Proposition \ref{prop:tech} that for some
  $\mathbf{\tilde r}=((\tilde r^0_{ij})_{1\le j \le w_i}, (\tilde
  r^1_{ij})_{2\le j \le w_i})_{i \in [d]}$ where each $\tilde
  r^0_{ij}, \tilde r^1_{ij}\in[\lambda \beta_c^ \Delta, \lambda
  \beta_c^{-\Delta}]$,
  \begin{align*}
    \abs{\phi(R_v)-\phi(R(G,v,L))}
    &\le\sum_{i=1}^d\tuple{\sum_{j=1}^{w_i}
      \frac{\Phi(f)}{\Phi(r^0_{ij})}\abs{\frac{\partial
          f(\mathbf{r})}{\partial r^0_{ij}}}\cdot\abs{\phi(r^0_{ij})-\phi(\hat r^0_{ij})} + \sum_{j = 2}^{w_i} \frac{\Phi(f)}{\Phi(r^1_{ij})}\abs{\frac{\partial f(\mathbf{r})}{\partial r^1_{ij}}}\cdot\abs{\phi(r^1_{ij})-\phi(\hat r^1_{ij})}} \\
    &\overset{(\spadesuit)}{\le}
    4\sqrt{e}\sum_{i=1}^d\tuple{\sum_{j=1}^{w_i}
      \frac{\Phi(f)}{\Phi(r^0_{ij})}\abs{\frac{\partial
          f(\mathbf{r})}{\partial r^0_{ij}}} + \sum_{j = 2}^{w_i} \frac{\Phi(f)}{\Phi(r^1_{ij})}\abs{\frac{\partial f(\mathbf{r})}{\partial r^1_{ij}}}}\alpha^{L-\lfloor 1+c\log_{1/\alpha} w_i\rfloor}\\
    &\overset{(\heartsuit)}{\le} 4\sqrt{e}\alpha^L.
  \end{align*}
  $(\spadesuit)$ follows from the induction hypothesis and
  $(\heartsuit)$ is due to Lemma \ref{lem:ising-decay} for all $d \le
  \Delta$.

  Then the lemma follows from Proposition \ref{prop:tech}, since
  \begin{align*}
    \abs{R_v-R(G,v,L)}
    &=\frac{1}{\Phi(\tilde x)}\cdot\abs{\phi(R_v)-\phi(R(G,v,L))}&\mbox{for some $\tilde x\in[\lambda \beta_c^ \Delta, \lambda\beta_c^{-\Delta}]$}\\
    &\le e^{\sqrt{e}}\lambda^{-1}\cdot 4\sqrt{e}\alpha^L\\
    &= 4\lambda^{-1}e^{\frac{1}{2} + \sqrt{e}}\alpha^L.
  \end{align*}
  The inequality above is due to $\beta_c^{-\Delta} = \tuple{1 -
    \frac{\sqrt{e}}{\Delta + \frac{3}{2}\sqrt{e}}}^{-\Delta} \le
  e^{\sqrt{e}}.$
\end{proof}

\end{document}